\documentclass{article}
\usepackage[utf8]{inputenc}
\usepackage{fullpage}
\usepackage{hyperref}
\usepackage{mathtools}
\DeclarePairedDelimiter\abs{\lvert}{\rvert}

\usepackage{amsfonts,amsmath}
\usepackage{soul}
\usepackage{amsthm}
\usepackage{mathrsfs}
\usepackage{authblk}

\DeclareMathOperator{\diag}{diag}

\newtheorem{thm}{Theorem}

\newtheorem{pro}{Proposition}
\newtheorem{lem}{Lemma}

\title{Local Dirac energy decay in the 5D Myers-Perry geometry using an integral spectral representation for the Dirac propagator}
\author{Qiu Shi Wang$^1$}
\affil{$^1$\small Department of Mathematics and Statistics\\ McGill University \\ Montréal, QC, H3A 2K6, Canada\\ qiu.s.wang@mail.mcgill.ca}
\date{\today}

\begin{document}

\maketitle
\begin{abstract}
We consider the massive Dirac equation in the exterior region of the 5-dimensional Myers-Perry black hole. Using the resulting ODEs obtained from the separation of variables of the Dirac equation, we construct an integral spectral representation for the solution of the Cauchy problem with compactly supported smooth initial data. We then prove that the probability of presence of a Dirac particle to be in any compact region of space decays to zero as $t\to\infty$, in analogy with the case of the Dirac operator in the Kerr-Newman geometry \cite{Finster_2003}.
\end{abstract}
\section{Introduction}

Chandrasekhar's separation of variables for the Dirac equation in the Kerr geometry \cite{chandrasekhar} provides an avenue for the study of the long-term behaviour of solutions of the massive Dirac equation in 4-dimensional black hole geometries. In \cite{Finster_2003}, Finster, Kamran, Smoller and Yau obtain decay estimates for spinor fields in the exterior region of the non-extreme Kerr-Newman geometry using an integral spectral representation that they established for solutions to the Cauchy problem for the Dirac equation. In \cite{daude:tel-00011974,AIF_2004__54_6_2021_0}, Daudé studies energy decay and scattering of Dirac spinors in a class of spacetimes that includes the Kerr-Newman geometry. In the latter geometry, Batic \cite{Batic_2007} develops a time-dependent scattering theory of Dirac particles using an integral representation for the Dirac propagator.

Similar questions may be asked concerning the Dirac equation in higher-dimensional black hole geometries; the present paper is concerned with the Myers-Perry metrics. Recall that the latter metrics, derived by Myers and Perry \cite{MYERS1986304}, are the $(d>4)$-dimensional generalizations of the Kerr metric, with rotations along $\left \lfloor{(d-1)/2}\right \rfloor$ independent axes. In particular, the 5D Myers-Perry black hole has two independent angular momenta.

The separation of variables for the Dirac equation in the Kerr-Newman geometry can be done equivalently using the Newman-Penrose formalism or the orthonormal local Lorentz frame formalism. In the absence of a suitable generalization of the Newman-Penrose formalism to 5 dimensions, Wu \cite{Wu_2008} chooses an orthonormal local Lorentz frame (pentad) to formulate the Dirac equation in the 5D Myers-Perry geometry and separate it into radial and angular ODEs. In the same geometry, Daudé and Kamran \cite{Daud__2012} use Wu's pentad and a ``locally non-rotating'' pentad that they construct to prove local energy decay for Dirac spinors using Mourre theory, thus generalizing the decay result of \cite{Finster_2003} to 5 dimensions. The objective of the present work is to obtain a similar result to that of \cite{Daud__2012}, namely the local decay of Dirac spinors in the 5D Myers-Perry geometry, by constructing an integral spectral representation for the solution to the Cauchy problem with compactly supported smooth initial data, in analogy with what was done in the Kerr-Newman geometry in \cite{Finster_2003}.

The paper is organized as follows. In Section \ref{prelims}, we introduce the 5D Myers-Perry metric and the Dirac equation in the orthonormal pentad formalism. In Section \ref{sepvar}, we perform separation of variables on the form of the Dirac equation obtained by \cite{Wu_2008,Daud__2012} and derive systems of radial and angular ODEs analogous to those obtained in \cite{Finster_2003} for the Kerr-Newman metric. In Section \ref{Hilbert}, we construct a suitable Hilbert space structure on the space of Dirac spinors, then derive in Section \ref{fourierintegral} an integral spectral representation (\ref{intrep}) for the Dirac propagator (Theorem \ref{intreptheorem}). The decay of Dirac spinors in any compact region (Theorem \ref{maintheorem}) follows from Theorem \ref{intreptheorem} and the Riemann-Lebesgue lemma.

Throughout the text, we will highlight similarities and differences between the results obtained for the Kerr and 5D Myers-Perry geometries.

\section{The 5D Myers-Perry metric and the Dirac equation}\label{prelims}

In this section, we introduce the 5D Myers-Perry geometry and formulate the Dirac equation on it using the orthonormal pentad formalism.

In Boyer-Lindquist coordinates $(t,r,\theta,\varphi,\psi)$, a 5D Myers-Perry black hole is represented by the manifold
\begin{equation}
\mathcal{M}=\mathbb{R}_t\times (0,\infty)_r \times (0,\frac{\pi}{2})_\theta \times (0,2\pi)_\varphi \times (0,2\pi)_\psi
\end{equation}
equipped with the Lorentzian metric

\begin{equation}\label{metric}
ds^2=-dt^2+\frac{\Sigma r^2}{\Delta}dr^2+\Sigma d\theta^2 + (r^2+a^2)\sin^2\theta d\varphi^2 + (r^2+b^2)\cos^2\theta d\psi^2 + \frac{\mu}{\Sigma}(dt-a\sin^2\theta d\varphi - b\cos^2\theta d\psi)^2,
\end{equation}
where $\mu/2$ is the mass of the black hole, $a$ and $b$ its two independent angular momenta, and 
\begin{equation}
\Delta=(r^2+a^2)(r^2+b^2)-\mu r^2,\qquad \Sigma=r^2+a^2\cos^2\theta+b^2\sin^2\theta.
\end{equation}

Furthermore, the 5D Myers-Perry metric has three commuting Killing vector fields $\partial_t$, $\partial_\varphi$ and $\partial_\psi$. We restrict our attention to the non-extreme case $\mu>a^2+b^2+2\abs{ab}$, for which $\Delta$ has two distinct positive roots

\begin{equation}
r_\pm^2 = \frac{1}{2}\big(\mu -a^2-b^2\pm \sqrt{(\mu-a^2-b^2)^2-4a^2b^2}\big).
\end{equation}

The radii $r_-,r_+$ are the Cauchy and event horizons respectively. We consider only the exterior region $r>r_+$.

We choose the gamma matrices $\gamma^A$, $A=0,1,2,3,5$, as
\begin{equation}
\gamma^0=i\begin{pmatrix}
0 & I \\
I & 0
\end{pmatrix},
\gamma^1=i\begin{pmatrix}
0 & \sigma^3 \\
-\sigma^3 & 0
\end{pmatrix},
\gamma^2=i\begin{pmatrix}
0 & \sigma^1\\
-\sigma^1 & 0
\end{pmatrix},
\gamma^3=i\begin{pmatrix}
0 & \sigma^2\\
-\sigma^2 & 0
\end{pmatrix},
\gamma^5=\begin{pmatrix}
I & 0\\
0 & -I
\end{pmatrix},
\end{equation}

\noindent where the $\sigma^j$ are the Pauli matrices. They satisfy the Clifford algebra anticommutation relations

\begin{equation}
\{\gamma^A,\gamma^B\}=2\eta^{AB},
\end{equation}

\noindent where $\eta^{AB}=\diag\{-1,1,1,1,1\}$ is the five-dimensional Minkowski metric in our chosen signature. We define the matrices $\Gamma^A$ by $\Gamma^0=i\gamma^0$ and $\Gamma^j=-\gamma^0\gamma^j$ for $j\neq 0$; they satisfy the anticommutation relations

\begin{equation}
\{\Gamma^A,\Gamma^B\}=2\delta^{AB},
\end{equation}

\noindent where $\delta^{AB}=\diag\{1,1,1,1,1\}$ is the Kronecker delta.

The massive Dirac equation takes the form

\begin{equation}
(\gamma^A(\partial_A+\Gamma_A)-m)\phi=0,
\end{equation}

\noindent where $\Gamma_A$ are the components of the spinor connection  $\Gamma=\Gamma_Ae^A=\frac{1}{4}\gamma^A\gamma^B\omega_{AB}$ in the orthonormal pentad frame $e^A={e^A}_\mu dx^\mu$, and $\omega_{AB}$ is the connection 1-form in the same frame.

\section{Separation of variables}\label{sepvar}

The Dirac equation on the 5D Myers-Perry metric is put into a separable form $(\mathcal{R}+\mathcal{A})\Psi=0$ using a local Lorentz frame in \cite{Wu_2008,Daud__2012}. In this section, we summarize this procedure, then perform separation of variables using an ansatz (\ref{ansatz}) analogous to the one used in \cite{Finster_2003}. The resulting angular and radial ODEs are similar to the corresponding results in the Kerr geometry \cite{Finster_2003}, but differ slightly due to the additional spatial dimension affecting the form of the horizon-defining function $\Delta$ and the $\theta$-singularities of the angular ODE.

Using the local Lorentz frame $\partial_A={e_A}^\mu\partial_\mu$ given by 

\begin{align}
\begin{split}
\partial_0&=\frac{(r^2+a^2)(r^2+b^2)}{r\sqrt{\Delta\Sigma}}\Big(\partial_t+\frac{a}{r^2+a^2}\partial_\varphi + \frac{b}{r^2+b^2}\partial_\psi\Big),\\
\partial_1&=\sqrt{\frac{\Delta}{r^2\Sigma}}\partial_r,\\
\partial_2&=\frac{1}{\sqrt{\Sigma}}\partial_\theta,\\
\partial_3&=\frac{\sin\theta\cos\theta}{p\sqrt{\Sigma}}\Big((a^2-b^2)\partial_t+\frac{a}{\sin^2\theta}\partial_\varphi-\frac{b}{\cos^2\theta}\partial_\psi\Big),\\
\partial_5&=\frac{1}{rp}\Big(ab\partial_t+b\partial_\varphi+a\partial_\psi\Big),
\end{split}
\end{align}

\noindent where $p^2=a^2\cos^2\theta+b^2\sin^2\theta$, making the substitution 
\begin{equation}
\Psi=\Delta^{1/4}\sqrt{r+ip\gamma^5}\phi = \left(\sqrt{\frac{r+\sqrt{\Sigma}}{2}}I+i\sqrt{\frac{\sqrt{\Sigma}-r}{2}}\gamma^5\right)\phi
\end{equation}
and introducing in the exterior region the radial variable $x$ defined by

\begin{equation}\label{radialvariable}
\frac{dx}{dr}=\frac{(r^2+a^2)(r^2+b^2)}{\Delta},
\end{equation}

\noindent one obtains the Hamiltonian form for the Dirac equation 
\begin{align}
i\partial_t \Psi&=\mathcal{D}\Psi,\\ \mathcal{D}&=N\mathcal{D}_0, 
\end{align}

\noindent where, writing $D_\alpha=-i\partial_\alpha$,

\begin{equation}
N^{-1}=I+\frac{r\sqrt{\Delta}}{(r^2+a^2)(r^2+b^2)}\Big( \frac{(a^2-b^2)}{p}\sin\theta\cos\theta\Gamma^3+\frac{ab}{p}\Gamma^5+\frac{ab}{r}\Gamma^0\Big),
\end{equation}

\begin{equation}
\mathcal{D}_0=\mathbb{D}_0+\frac{\sqrt{\Delta}}{(r^2+a^2)(r^2+b^2)}\left( b\Gamma^0 D_\varphi + a\Gamma^0 D_\psi + \frac{ab}{r}\gamma^1+mpr\Gamma^5\right),
\end{equation}

\begin{equation}
\mathbb{D}_0=\Gamma^1D_x+\frac{r\sqrt{\Delta}}{(r^2+a^2)(r^2+b^2)}\mathbb{D}_{S^3}-m\frac{r^2\sqrt{\Delta}}{(r^2+a^2)(r^2+b^2)}\Gamma^0+\frac{a}{r^2+a^2}D_\varphi + \frac{b}{r^2+b^2}D_\psi,
\end{equation}

\noindent and

\begin{equation}\label{DS3}
\mathbb{D}_{S^3}=i\gamma^0 \gamma^2 \Big( \partial_\theta + \frac{\cot\theta}{2}-\frac{\tan\theta}{2}\Big) + i\gamma^0\gamma^3 \frac{1}{\sin\theta}\partial_\varphi+i\gamma^0\gamma^5\frac{1}{\cos\theta}\partial_\psi.    
\end{equation}

It is shown in \cite[Appendix A]{Daud__2012} that the matrix $N^{-1}$ is positive definite. 

We use a separation ansatz analogous to the one used in \cite{Finster_2003}, but with two spinor angular momentum parameters $k_a,k_b\in\mathbb{N}$; it is of the form

\begin{equation}\label{ansatz}
\Psi(t,x,\theta,\varphi,\psi)=e^{-i\omega t}e^{-i((k_a+\frac{1}{2})\varphi+(k_b+\frac{1}{2})\psi)}\begin{pmatrix}
X_+(x) Y_+(\theta)\\
X_-(x) Y_-(\theta)\\
X_-(x) Y_+(\theta)\\
X_+(x) Y_-(\theta)
\end{pmatrix}.
\end{equation}

The Dirac equation can thus be written in the form

\begin{equation}\label{diracRA}
(\mathcal{R}+\mathcal{A})\Psi=0,
\end{equation}
where
\begin{multline}
\mathcal{R}=\frac{(r^2+a^2)(r^2+b^2)}{r\sqrt{\Delta}}\Gamma^1 D_x -mr\Gamma^0 + \Big( \frac{a(r^2+b^2)}{r\sqrt{\Delta}}+\frac{a}{r}\Gamma^0\Big) D_\varphi \\
+\Big(\frac{b(r^2+a^2)}{r\sqrt{\Delta}}+\frac{b}{r}\Gamma^0\Big) D_\psi + \frac{ab}{r^2}\gamma^1-\omega\Big( \frac{(r^2+a^2)(r^2+b^2)}{r\sqrt{\Delta}}+\frac{ab}{r}\Gamma^0\Big)
\end{multline}
and
\begin{equation}\label{angularoperator}
\mathcal{A}=\mathbb{D}_{S^3}+mp\Gamma^5-\omega\Big(\frac{(a^2-b^2)\sin\theta\cos\theta}{p}\Gamma^3+\frac{ab}{p}\Gamma^5\Big)
\end{equation}
are purely radial and angular operators respectively, with $\mathbb{D}_{S^3}$ as in (\ref{DS3}). 

The eigenvalue relations $D_\varphi\Psi=-(k_a+\frac{1}{2})\Psi$ and $D_\psi\Psi=-(k_b+\frac{1}{2})\Psi$ follow from the ansatz (\ref{ansatz}), under which the Dirac equation (\ref{diracRA}) separates into ODEs $-\mathcal{R}\Psi=\mathcal{A}\Psi=\lambda\gamma^1\Psi$ for a separation constant $\lambda$. In terms of $X=(X_+,X_-)$, the radial ODE is

\begin{equation}\label{radial}
\left[\frac{d}{dx}+i\Omega(x)\begin{pmatrix}
1 & 0\\
0 & -1
\end{pmatrix}
\right] X = \frac{\sqrt{\Delta}}{(r^2+a^2)(r^2+b^2)}\begin{pmatrix}
0 & -ab/r-\lambda r + iB(x)\\
-ab/r-\lambda r - iB(x) & 0
\end{pmatrix}X,
\end{equation}

\noindent where

\begin{align}
\Omega(x)&=\omega + \frac{a(k_a+\frac{1}{2})}{r^2+a^2}+\frac{b(k_b+\frac{1}{2})}{r^2+b^2},\\
B(x)&=mr^2+a\left(k_a+\frac{1}{2}\right)+b\left(k_b+\frac{1}{2}\right)+\omega ab,
\end{align}

\noindent recalling that $r$ is defined as function of $x$ by (\ref{radialvariable}). In terms of $Y=(Y_+,Y_-)$, the angular ODE is

\begin{equation}\label{angular}
AY\equiv
\begin{pmatrix}
C_b(\theta) & L_\theta+C_a(\theta)\\
-L_\theta + C_a(\theta) & -C_b(\theta)
\end{pmatrix}
Y = \lambda Y,
\end{equation}

\noindent where
\begin{align}\label{angularabbrevs}
\begin{split}
L_\theta&=\partial_\theta + \frac{\cot\theta}{2}-\frac{\tan\theta}{2},\\
C_a(\theta)&=-\frac{(k_a+\frac{1}{2})}{\sin\theta}-\frac{\omega(a^2-b^2)\sin\theta\cos\theta}{p},\\
C_b(\theta)&=-\frac{(k_b+\frac{1}{2})}{\cos\theta}+mp-\frac{\omega ab}{p}.
\end{split}
\end{align}

The radial and angular ODEs resemble closely those obtained from the separation of variables for the Dirac equation in the Kerr metric \cite{Finster_2003}, which, in Boyer-Lindquist coordinates $(t,r,\theta,\varphi)$ with a change of radial coordinate $du/dr=(r^2+a^2)/\Delta_\mathrm{Kerr}$, are of the form

\begin{align}\label{kerrradial}
\left[\frac{d}{du}+i\Omega_\mathrm{Kerr}(u)\begin{pmatrix}
1 & 0\\
0 & -1
\end{pmatrix}
\right] X_\mathrm{Kerr} &=\frac{\Delta_\mathrm{Kerr}}{r^2+a^2}\begin{pmatrix}
0 & imr-\lambda\\
-imr-\lambda & 0
\end{pmatrix}X_\mathrm{Kerr},\\
\begin{pmatrix}
-am\cos\theta & \mathcal{L}_-\\
-\mathcal{L}_+ & am\cos\theta
\end{pmatrix}Y_\mathrm{Kerr}&=\lambda Y_\mathrm{Kerr}.\label{kerrangular}
\end{align}

\noindent where $\Delta_\mathrm{Kerr}=r^2-\mu r-a^2$ is the horizon-defining function for the Kerr metric,

\begin{align}
\Omega_\mathrm{Kerr}(u)&=\omega + \frac{a(k+\frac{1}{2})}{r^2+a^2},\\
\mathcal{L}_\pm &= \partial_\theta + \frac{\cot\theta}{2}\mp\left(a\omega\sin\theta+\frac{k+\frac{1}{2}}{\sin\theta}\right),
\end{align}
\noindent and $k\in\mathbb{Z}$ the azimuthal quantum number. In the limiting case of a 5D Myers-Perry black hole with zero angular momentum about one of the two independent axes of rotation, say the $\partial_\psi$ direction so that $b=0$, (\ref{radial}) and (\ref{angular}) simplify to

\begin{multline}\label{b=0radial}
\left[\frac{d}{dx}+i\Omega_\mathrm{Kerr}(x)\begin{pmatrix}
1 & 0\\
0 & -1
\end{pmatrix}
\right] X=\\
\frac{\sqrt{r^2+a^2-\mu}}{r^2+a^2}\begin{pmatrix}
0 & imr-\lambda + i(a/r)(k_a+\frac{1}{2}) \\
-imr-\lambda - i(a/r)(k_a+\frac{1}{2}) & 0
\end{pmatrix}X,
\end{multline}
\begin{equation}\label{b=0angular}
\begin{pmatrix}
am\cos\theta-\frac{k_b+1/2}{\cos\theta} & L_\theta - a\omega\sin\theta - \frac{k_a+1/2}{\sin\theta}\\
-L_\theta - a\omega\sin\theta - \frac{k_a+1/2}{\sin\theta} & -am\cos\theta+\frac{k_b+1/2}{\cos\theta}
\end{pmatrix}Y=\lambda Y.
\end{equation}

The $b=0$ radial equation (\ref{b=0radial}) differs from the Kerr result (\ref{kerrradial}) only by the different $\mu$-dependance and the addition of a $1/r$-damped angular term. Up to some signs, the $b=0$ angular equation (\ref{b=0angular}) differs from (\ref{kerrangular}) only by the addition of $\sim \sec\theta$ terms singular at $\theta=\pi/2$ arising from the Hopf coordinates $(\theta,\varphi,\psi)$ on the $S^3$ geometry of constant $t,r$ hypersurfaces of the 5D Myers-Perry black hole.

\section{Hilbert space structure}\label{Hilbert}
In this section, we construct suitable scalar products $(\cdot|\cdot)$ and $\langle\cdot|\cdot\rangle$ on the Dirac spinors with respect to which $\gamma^1\mathcal{A}$ and $\mathcal{D}$ are respectively essentially self-adjoint, then describe suitable radial boundary conditions to impose to preserve self-adjointness of the Hamiltonian on a radial interval $I\subset \mathbb{R}_x$ bounded on one or both sides.

Let $\Sigma_t$ be a constant-$t$ hypersurface $\mathbb{R}_x \times S^3_{\theta,\varphi,\psi}$. Let $\Psi, \Phi \in L^2(\Sigma_t, d\nu)$, where $d\nu=dxd\Omega$ is the product of Lebesgue measure on $\mathbb{R}_x$ and the usual measure $d\Omega=\sin\theta\cos\theta d\theta d\varphi d\psi$ on $S^3$. Designating by $\overline{\Psi}$ the conjugate transpose spinor (not the adjoint spinor $\Psi^\dagger=\overline{\Psi}\gamma^0$), we define the scalar product

\begin{equation}\label{scalar1}
\big(\Psi|\Phi\big)=\int_{\Sigma_t} \overline{\Psi}\Phi\: d\nu.
\end{equation}

If we define $Z^{k_ak_b\omega\lambda}=\exp\left(-i\left((k_a+\frac{1}{2})\varphi+(k_b+\frac{1}{2})\psi\right)\right)Y^{k_ak_b\omega\lambda}$ and the scalar products

\begin{align}
\big(X^{k_ak_b\omega\lambda}|X^{k_a'k_b'\omega'\lambda'}\big)&=\int_{-\infty}^\infty \overline{X}^{k_ak_b\omega\lambda}(x)X^{k_a'k_b'\omega'\lambda'}(x) \: dx,\label{xscalarproduct}\\
\big(Y^{k_ak_b\omega\lambda}|Y^{k_a'k_b'\omega'\lambda'}\big)&=\int_{S^3} \overline{Z}^{k_ak_b\omega\lambda}Z^{k_a'k_b'\omega'\lambda'}\: d\Omega,
\end{align}

\noindent then (\ref{scalar1}) factors in the sense that

\begin{equation}
\big(\Psi^{k_ak_b\omega\lambda}|\Psi^{k_a'k_b'\omega'\lambda'}\big)=\big(X^{k_ak_b\omega\lambda}|X^{k_a'k_b'\omega'\lambda'}\big)\big(Y^{k_ak_b\omega\lambda}|Y^{k_a'k_b'\omega'\lambda'}\big).
\end{equation}

Let $\mathcal{G}$ be the Hilbert space obtained by completion of the space of spinors of the form (\ref{ansatz}) with respect to the scalar product (\ref{scalar1}). We have that $\mathcal{D}_0$ is essentially self-adjoint on $\mathcal{G}$, but not $\mathcal{D}=N\mathcal{D}_0$. To compensate for the factor $N$, we therefore define the scalar product

\begin{equation}\label{scalar2}
\langle \Psi | \Phi \rangle = \int_{\Sigma_t} \overline{\Psi} N^{-1} \Phi\: d\nu,
\end{equation}

\noindent and let $\mathcal{H}$ be the Hilbert space of spinors in $L^2(\mathbb{R}\times S^3,d\nu)^4$ equipped with scalar product (\ref{scalar2}). In order to construct an integral spectral representation for $e^{-it\mathcal{D}}$, it is important that $\mathcal{D}$ be essentially self-adjoint with respect to (\ref{scalar2}). This is precisely what was shown in Lemma 19 of \cite[Appendix A]{Daud__2012}, namely

\begin{lem}[\cite{Daud__2012}]
The Hamiltonian $\mathcal{D}$ is essentially self-adjoint on $\mathcal{H}$ with domain

\begin{equation}
D(\mathcal{D})=\{\psi\in\mathcal{G}, \lVert \mathcal{D}_0 \psi \rVert^2 <\infty\}.    
\end{equation}
\end{lem}

Note that the space of smooth, compactly supported test spinors $C^\infty_0(\mathbb{R}\times S^3)^4$ is a dense subset of the domain $D(\mathcal{D})$. Another possible method to prove the essential self-adjointness of $\mathcal{D}$ could be to first show its symmetry with respect to (\ref{scalar2}), then to use Chernoff's method \cite{CHERNOFF1973401}, similarly to what was done in \cite{Finster_2016}.

The scalar product (\ref{scalar2}) does not factor into a product. More precisely, we have 

\begin{multline}\label{nonorthogonality}
\langle \Psi | \Psi' \rangle = \big( \Psi | \Psi' \big) - (a^2-b^2)\big(X|\frac{r\sqrt{\Delta}}{(r^2+a^2)(r^2+b^2)}\sigma^2|X'\big)\big(Y|\frac{\sin\theta\cos\theta}{p}\sigma^1|Y'\big)\\
+ab\big(X|\frac{r\sqrt{\Delta}}{(r^2+a^2)(r^2+b^2)}\sigma^2|X'\big)\big(Y|\frac{1}{p}\sigma^3|Y'\big)-ab\big(X|\frac{\sqrt{\Delta}}{(r^2+a^2)(r^2+b^2)}\sigma^2|X'\big)\big(Y|Y'\big).
\end{multline}

The ``non-orthogonality relation'' (\ref{nonorthogonality}) between the two scalar products is more complicated than in the Kerr geometry, but the presence of the $\sqrt{\Delta}$ factor in the non-orthogonal terms implies that they decay exponentially near the event horizon, as in the Kerr case; this behaviour is required to construct the integral representation formula (\ref{Dx2}) for the Dirac propagator.

As in \cite{Finster_2003}, suitably modifying the bounds of the integral (\ref{xscalarproduct}), we may restrict the Hilbert space in the radial direction $\mathbb{R}_x$ to intervals $(-\infty,x_2]$ or $[x_1,x_2]$; we denote the respective restricted Hamiltonians by $\mathcal{D}_{x_2}$ or $\mathcal{D}_{x_1,x_2}$. Self-adjointness of the Hamiltonian is preserved provided we impose at the boundaries $x_j$ the boundary conditions 
\begin{equation}\label{bc}
X_+(x_j)=X_-(x_j).
\end{equation}
\section{Integral spectral representation and decay estimates}\label{fourierintegral}

In this section, we use the radial ODE (\ref{radial}) and the Hilbert space structure obtained above to obtain an integral representation formula for the Dirac propagator and thus conclude $t\to\infty$ decay of the probability of the Dirac particle to be in any compact region of space. As the construction follows the Kerr-Newman case closely, we will only describe the ways in which the proofs differ from those in \cite{Finster_2003}, and omit or provide sketches of the identical arguments.

First, we construct a suitable basis for $\mathcal{H}_{x_1,x_2}$ of eigenvectors of $\mathcal{D}_{x_1,x_2}$. We choose the basis vectors to be eigenvectors of $D_\varphi$ and $D_\Psi$ with eigenvalues $k_a+\frac{1}{2}$ and $k_b+\frac{1}{2}$ respectively. Denote the restriction of $\mathcal{D}_{x_1,x_2}$ to a fixed $k_a,k_b$ subspace by $\mathcal{D}^{k_ak_b}_{x_1,x_2}$, and the restriction of $\mathcal{A}$ to $\mathcal{H}_{x_1x_2}^{k_ak_b}$ by $\mathcal{A}^{k_ak_b}$.  

Then, we choose them to be solutions of the angular ODE $AY=\lambda Y$, or equivalently we require them to be eigensolutions of $\gamma^1\mathcal{A}^{k_ak_b}\Psi=\lambda\Psi$ of the form (\ref{ansatz}). By the nondegeneracy and regularity of the spectrum of $A$ as shown in the Appendix, we may write the angular eigenvalues as smooth functions $\lambda_n(\omega)$ with $\lambda_n<\lambda_{n+1}$ for all $\omega\in\mathbb{R}$. For fixed $k_a, k_b, \omega$, denote by $N(k_a,k_b,\omega)$ the set of $n\in\mathbb{Z}$ for which the radial equation (\ref{radial}) with $\lambda=\lambda_n$ has a solution. Explicitly, our chosen basis is

\begin{equation}\label{chosenbasis}
\{\Psi^{k_ak_b\omega n}_{x_1,x_2}| k_a,k_b\in \mathbb{Z}, \omega\in \sigma(\mathcal{D}^{k_ak_b}_{x_1,x_2}), n\in N(k_a,k_b,\omega)\},
\end{equation}
normalized such that $(X^{k_ak_b\omega n}_{x_1,x_2}|X^{k_ak_b\omega n}_{x_1,x_2})=1$ and $(Y^{k_ak_b\omega n}|Y^{k_ak_b\omega n})=1$.

The following result is analogous to Lemma 3.1 of \cite{Finster_2003} and describes the behaviour of the spinors near the event horizon $r_+$.

\begin{lem}\label{lem1}
Every nontrivial solution $X$ of (\ref{radial}) with boundary conditions (\ref{bc}) is asymptotically as $x\rightarrow -\infty$ of the form

\begin{equation}\label{oscillatory}
X(x)=\begin{pmatrix}
e^{-i\Omega_0x}f_0^+\\
e^{i\Omega_0x}f_0^-
\end{pmatrix}
+R_0(x)
\end{equation}

\noindent with

\begin{align}
|f_0^+|^2+|f_0^-|^2&\neq 0,\\
\Omega_0&=\omega + \frac{a(k_a+\frac{1}{2})}{r_+^2+a^2}+\frac{b(k_b+\frac{1}{2})}{r_+^2+b^2},\\
|R_0|&\leq ce^{dx},
\end{align}
for suitable $c,d>0$ which can be chosen locally uniformly in $\omega$.
\end{lem}

\begin{proof}
Writing $f=(f^+,f^-)$ and making the ansatz

\begin{equation}
X(x)=\begin{pmatrix}
e^{-i\Omega_0x}f^+(x)\\
e^{i\Omega_0x}f^-(x)
\end{pmatrix},
\end{equation}

\noindent one obtains

\begin{multline}\label{fode}
f'=\biggl( i(\Omega_0-\Omega(x))\begin{pmatrix}
1 & 0\\
0 & -1
\end{pmatrix}
+\\
\frac{\sqrt{\Delta}}{(r^2+a^2)(r^2+b^2)}\begin{pmatrix}
0 & (-ab/r-\lambda r + iB(x))e^{2i\Omega_0x}\\
(-ab/r-\lambda r - iB(x))e^{-2i\Omega_0x} & 0
\end{pmatrix}
\biggr) f.
\end{multline}

The prefactor of the right-hand side of (\ref{fode}) tends exponentially to zero as $x\rightarrow -\infty$, so the result follows by the same argument as in \cite{Finster_2003}.
\end{proof}

We thus make the normalization 
\begin{equation}\label{Xx2normalization}
\lim_{x\to-\infty} |X_{x_2}|=1.
\end{equation}
Then we have, analogously to Lemma 3.2 of \cite{Finster_2003},

\begin{lem}\label{lem2}
For fixed $x_2$ and asymptotically as $x_1\rightarrow-\infty$, 

\begin{equation}\label{grel}
X_{x_1,x_2}=g(x_1)X_{x_2}|_{[x_1,x_2]}
\end{equation}
with
\begin{equation}\label{gdef}
|g(x_1)|^{-2}=(x_2-x_1)+\mathcal{O}(1).
\end{equation}

Furthermore, 
\begin{equation}\label{orth}
\abs{\langle \Psi^{k_ak_b\omega n}_{x_1,x_2}|\Psi^{k_ak_b\omega n'}_{x_1,x_2}\rangle - \delta^{nn'}}\leq \frac{c}{x_2-x_1}\sum_{l=1}^3\abs{\langle Y^{k_ak_b\omega n}|J_l|Y^{k_ak_b\omega n'}\rangle},
\end{equation}

\noindent where

\begin{equation}\label{Jj}
J_1=\frac{\sin\theta\cos\theta}{p}(b^2-a^2)\sigma^1,\qquad J_2=\frac{ab}{p}\sigma^3,\qquad J_3=-\frac{ab}{r_+}I.
\end{equation}
\end{lem}

\begin{proof}
The proof of (\ref{grel}) and (\ref{gdef}) is identical to that in \cite{Finster_2003}. The estimate (\ref{orth}) and (\ref{Jj}) follow from the proof in \cite{Finster_2003}, the triangle inequality, and the fact that $r>r_+$ in the region of interest, which implies that

\begin{equation}
\int_{x_1}^{x_2}\frac{r\sqrt{\Delta}}{(r^2+a^2)(r^2+b^2)}\abs{X}\abs{X'}\: dx\leq \frac{1}{r_+}\int_{x_1}^{x_2}\frac{\sqrt{\Delta}}{(r^2+a^2)(r^2+b^2)}\abs{X}\abs{X'}\: dx.
\end{equation}
\end{proof}

\begin{lem}\label{gapestimate}
For fixed $x_2$ and asymptotically as $x_1\rightarrow -\infty$, 

\begin{equation}
\Delta\omega_{k_ak_bn}=\frac{\pi}{x_2-x_1} + \mathcal{O}((x_2-x_1)^{-2}),
\end{equation}
locally uniformly in $\omega$.
\end{lem}

\begin{proof}
The matrix on the right-hand side of (\ref{radial}) is Hermitian, and therefore $\abs{X_+}=\abs{X_-}$ for all $x\leq x_2$, as in \cite{Finster_2003}. Writing (\ref{fode}) as $\frac{df}{dx}=Gf$, we have

\begin{equation}
\abs{\frac{d}{dx}\partial_\omega f}\leq \abs{(\partial_\omega G)f}+\abs{G(\partial_\omega f)}.
\end{equation}
Noting that $\partial_\omega B(x) = ab$ and $\partial_\omega \Omega(x)=\partial_\omega\Omega_0=1$, we obtain
\begin{equation}
\partial_\omega G = \frac{\sqrt{\Delta}}{(r^2+a^2)(r^2+b^2)}\begin{pmatrix}
0 & G_+ e^{2i\Omega_0x}\\
G_- e^{-2i\Omega_0x} & 0
\end{pmatrix},
\end{equation}
where $G_\pm=-2xB(x)\mp 2ix(ab/r+\lambda r - ab/(2x))$. Therefore, there are $c_1, d_1>0$ such that $\abs{\partial_\omega G}\leq c_1\abs{x}e^{d_1x}$ for all $x\leq x_2$, which in turn implies that there is a $c_2>0$ such that $\abs{\partial_\omega G}\leq c_2 e^{d_1x/2}$. Using the fact that there are $c_3, d_3>0$ such that $\abs{G}\leq c_3 e^{d_3x}$, we conclude, choosing $d=\min\{d_3, d_1/2\}$ for instance, that there exist $c_5,c_6>0$ such that

\begin{equation}
\abs{\frac{d}{dx}\partial_\omega f}\leq c_5 e^{dx} \abs{\partial_\omega f} + c_6 e^{dx} \abs{f}.
\end{equation}
The rest of the proof follows \cite{Finster_2003} exactly.
\end{proof}

The integral representation formula for $\mathcal{D}_{x_2}$ follows from the above estimates, noting that the bound in (\ref{orth}) is uniform in $k_a,k_b$ and $\omega$. The proof follows \cite{Finster_2003} closely.

\begin{pro}
For every $\Psi\in C^\infty_0((-\infty, x_2]\times S^3)^4$ and $y=(x,\theta,\varphi,\psi)$,

\begin{equation}\label{Dx2}
\left( e^{-it\mathcal{D}_{x_2}}\Psi\right)(y)=\frac{1}{\pi}\sum_{k_a,k_b\in \mathbb{Z}} \int_{-\infty}^\infty d\omega\: e^{-i\omega t} \sum_{n\in\mathbb{Z}} \Psi^{k_ak_b\omega n}_{x_2}(y)\langle \Psi^{k_ak_b\omega n}_{x_2}|\Psi\rangle.
\end{equation}
\end{pro}

\begin{proof}
The Dirac propagator $\exp(-it\mathcal{D}_{x_1,x_2})$ can be expressed in terms of the chosen eigenbasis (\ref{chosenbasis}) as

\begin{equation}
e^{-it\mathcal{D}_{x_1,x_2}}\Psi = \sum_{k_a,k_b\in\mathbb{Z}}\sum_{\omega\in\sigma(\mathcal{D}^{k_ak_b}_{x_1,x_2})}e^{-i\omega t}\sum_{n,n'\in N(k_a,k_b,\omega)}c_{nn'}\Psi^{k_ak_b\omega n}_{x_1,x_2} \langle \Psi^{k_ak_b\omega n'}_{x_1,x_2}|\Psi\rangle,
\end{equation}

\noindent for some constants $c_{nn'}$ arising due to the fact that eigenvectors $\Psi^{k_ak_b\omega n}_{x_1,x_2}$ with different values of $n$ are not orthogonal. Nonetheless, as $x_1\rightarrow -\infty$, we have that $c_{nn'}\rightarrow \delta_{nn'}$ by (\ref{orth}). Expressing $\Psi_{x_1,x_2}$ in terms of $\Psi_{x_2}$ using (\ref{grel}) and (\ref{gdef}), we obtain

\begin{multline}\label{spectralsumProp1}
e^{-it\mathcal{D}_{x_1,x_2}}\Psi=\sum_{k_a,k_b\in\mathbb{Z}} \frac{1}{x_2-x_1} \sum_{\omega\in\sigma(\mathcal{D}^{k_ak_b}_{x_1,x_2})} e^{-i\omega t} \\
\times \left(\sum_{n\in N(k_a,k_b,\omega)} \Psi^{k_ak_b\omega n}_{x_2}\langle \Psi^{k_ak_b\omega n}_{x_2}|\Psi\rangle_{x_1,x_2} + \mathcal{O}((x_2-x_1)^{-1})\right),
\end{multline}
\noindent where the subscript on $\langle \cdot | \cdot \rangle$ indicates restriction of the inner product to the radial interval $[x_1,x_2]$, as explained at the end of Section \ref{Hilbert}. As $x_1\rightarrow -\infty$, Lemma \ref{gapestimate} implies that (\ref{spectralsumProp1}) converges to the Riemann integral (\ref{Dx2}).
\end{proof}

In the large $x$ regime, or equivalent the large $r$ regime since $\lim_{x\to \infty} x/r = 1$, the radial ODE (\ref{radial}) takes the asymptotic form

\begin{equation}\label{largex}
\frac{d}{dx}X=\left(\begin{pmatrix}
-i\omega & im \\
-im & i\omega
\end{pmatrix}
+\frac{1}{x}\begin{pmatrix}
0 & -\lambda \\
-\lambda & 0
\end{pmatrix}
\right) X + \mathcal{O}(x^{-2})X.
\end{equation}

We remark that (\ref{largex}) is in fact strictly simpler than the analogous equation in the Kerr geometry, namely

\begin{equation}\label{largeradiuskerr}
\frac{d}{du}X_\mathrm{Kerr}=\left(\begin{pmatrix}
-i\omega & im \\
-im & i\omega
\end{pmatrix}
+\frac{1}{u}\begin{pmatrix}
0 & -im\mu/2-\lambda \\
im\mu/2-\lambda & 0
\end{pmatrix}
\right) X_\mathrm{Kerr} + \mathcal{O}(u^{-2})X_\mathrm{Kerr}.
\end{equation}

Unlike in the Kerr geometry, the asymptotic behaviour of Dirac spinors far away from the 5D Myers-Perry black hole does not depend on the mass $\mu/2$ of the black hole. This is because the $\mu$-dependent term $\mp im\mu/2$ in (\ref{largeradiuskerr}) arises from the limit

\begin{equation}
\pm im\lim_{r\to\infty}(\sqrt{\Delta_\mathrm{Kerr}}-r)=\mp im\mu/2,
\end{equation}

\noindent while in the 5D Myers-Perry case, the relevant limit is

\begin{equation}
\pm im \lim_{r\to\infty}\left(\frac{\sqrt{\Delta}}{r}-r\right)=0.
\end{equation}

This is ultimately because the $\mu$-dependent term of $\Delta(r)$ is two orders in $r$ below the leading term $r^4$, while in $\Delta_\mathrm{Kerr}$ it is one order below the leading term $r^2$.

As in \cite{Finster_2003}, for $\abs{\omega}<\abs{m}$ we have exponentially growing/decaying fundamental solutions $\Psi_{1/2}$ which we normalize by 

\begin{equation}
\lim_{x\to\infty} \abs{\Psi_{1/2}(x)}=1.
\end{equation}

For $\abs{\omega}>\abs{m}$, normalize the two oscillating fundamental solutions $\Psi_{1/2}$ by

\begin{equation}\label{phaseconditions}
f_{0,1}=\begin{pmatrix}
1\\
0
\end{pmatrix},\qquad
f_{0,2}=\begin{pmatrix}
0\\
1
\end{pmatrix}
\end{equation}
with $f_{0,1/2}$ as in (\ref{oscillatory}). We have, simplified from \cite{Finster_2003},

\begin{lem}
Every nontrivial solution $X$ of (\ref{radial}) has for large $x$ the asymptotic form

\begin{equation}\label{finfty}
X(x)=A\begin{pmatrix}
e^{-i\Phi(x)}f_\infty^+\\
e^{i\Phi(x)}f_\infty^-
\end{pmatrix}
+R_\infty(x),
\end{equation}
with, for some constant $C>0$,
\begin{align}
\abs{f_\infty^+}^2+\abs{f_\infty^-}^2&\neq 0,\\
\Phi(x)&=\mathrm{sgn}(\omega)\sqrt{\omega^2-m^2}\:x,\\
A&=\begin{pmatrix}
\cosh\Theta & \sinh\Theta\\
\sinh\Theta & \cosh\Theta
\end{pmatrix}, \qquad \Theta=\frac{1}{4}\log\left(\frac{\omega+m}{\omega-m}\right),\\
\abs{R_\infty}&\leq \frac{C}{x}.
\end{align}
\end{lem}

Define the coefficients $f_\infty=(f_\infty^+,f_\infty^-)$ corresponding to the fundamental solutions $\Psi_{1/2}$ to be the \textit{transmission coefficients} $f_{\infty 1/2}$. As in \cite{Finster_2003},

\begin{thm}\label{intreptheorem}
For every $\Psi\in C_0^\infty(\mathbb{R}\times S^3)^4$,

\begin{equation}\label{intrep}
\left( e^{-it\mathcal{D}}\Psi\right)(y)=\frac{1}{\pi}\sum_{k_a,k_b,n\in\mathbb{Z}} \int_{-\infty}^\infty d\omega \: e^{-i\omega t} \sum_{a,b=1}^2t_{ab}^{k_ak_b\omega n} \Psi^{k_ak_b\omega n}_a(y) \langle \Psi^{k_ak_b\omega n}_b | \Psi \rangle,
\end{equation}
where for $\abs{\omega}<\abs{m}$ we have $t_{ab}=\delta_{a1}\delta_{b1}$, and for $\abs{\omega}>\abs{m}$,

\begin{equation}\label{taboscillatory}
t_{ab}=\frac{1}{2\pi}\int_0^{2\pi} \frac{t_a\overline{t_b}}{\abs{t_1}^2+\abs{t_2}^2} d\alpha,
\end{equation}
where the functions $t_a$ are given in terms of the transmission coefficients by
\begin{equation}\label{t1t2}
t_1(\alpha)=f_{\infty 2}^+ e^{-i\alpha} - f_{\infty 2}^- e^{i\alpha}, \qquad t_2(\alpha)=-f_{\infty 1}^+ e^{-i\alpha} + f_{\infty 1}^- e^{i\alpha}.
\end{equation}

The integral and series in (\ref{intrep}) converge weakly in $\mathcal{H}$.
\end{thm}

\begin{proof}
The proof follows along the lines of \cite{Finster_2003}, so we will only provide a brief motivation for the form of the integral representation (\ref{intrep}).

By the one-ended integral representation (\ref{Dx2}), at $t=0$ we have the ``completeness relation''

\begin{equation}\label{completenessrelation}
\Psi(y)=\frac{1}{\pi}\sum_{k_a,k_b\in\mathbb{Z}}\int_{-\infty}^\infty \: d\omega \sum_{n\in\mathbb{Z}} \Psi^{k_ak_b\omega n}_{x_2}(y)\langle \Psi^{k_ak_b\omega n}|\Psi\rangle.
\end{equation}

Since (\ref{completenessrelation}) is true for each value of $x_2$, we may take an average over $[x_2,x_2+T]$ and use Fubini's theorem to obtain

\begin{equation}\label{Taverage}
\Psi(y)=\frac{1}{\pi}\sum_{k_a,k_b\in\mathbb{Z}}\int_{-\infty}^\infty \: d\omega \sum_{n\in\mathbb{Z}}\left( \frac{1}{T}\int_0^T\: d\tau \Psi^{k_ak_b\omega n}_{x_2+\tau}(y)\langle \Psi^{k_ak_b\omega n}_{x_2+\tau}|\Psi\rangle\right).
\end{equation}

Expressing $\Psi^{k_ak_b\omega n}_{x_2+\tau}$ as a linear combination of the fundamental solutions $\Psi^{k_ak_b\omega n}_{1/2}$ by

\begin{equation}
\Psi^{k_ak_b\omega n}_{x_2+\tau}(y)=c_1(\tau)\Psi^{k_ak_b\omega n}_1 + c_2(\tau)\Psi^{k_ak_b\omega n}_2
\end{equation}

\noindent and choosing $c_1,c_2$ so that the boundary conditions (\ref{bc}) are satisfied, the expression in parentheses in (\ref{Taverage}) can be written as, suppressing some indices for clarity,

\begin{equation}\label{tabdefinition}
\frac{1}{T}\int_0^T d\tau \: \Psi_{x_2+\tau}(y)\langle \Psi_{x_2+\tau}|\Psi\rangle=\sum_{a,b=1}^2 t_{ab}(T) \Psi_a\langle\Psi_b|\Psi\rangle,
\end{equation}
\noindent where
\begin{equation}
t_{ab}(T)=\frac{1}{T}\int_0^T\: c_a(\tau) \overline{c_b(\tau)} \: d\tau.
\end{equation}

For $\abs{\omega}<\abs{m}$, in order for the spinor to satisfy the boundary conditions (\ref{bc}), the coefficient $c_2(\tau)$ of the exponentially growing fundamental solution $\Psi_2$ must tend to zero exponentially as $\tau\rightarrow \infty$, so, taking $T\rightarrow \infty$, we obtain $t_{ab}=\delta_{a1}\delta_{b1}$.

For $\abs{\omega}>\abs{m}$, the boundary conditions (\ref{bc}) and the normalization conditions (\ref{Xx2normalization}), (\ref{phaseconditions}), it can be shown that as $T\rightarrow \infty$ the coefficients $t_{ab}$ take the form (\ref{taboscillatory}) and (\ref{t1t2}); for further details, we refer the reader to\cite{Finster_2003}.

Taking the limit $T\rightarrow \infty$ in (\ref{Taverage}) and using (\ref{tabdefinition}), we obtain the integral representation (\ref{intrep}). The technical details justifying taking the limit $T\rightarrow \infty$ inside the series and the integral and showing that the resulting series and integral converge weakly in $\mathcal{H}$ can be found in \cite{Finster_2003}.

\end{proof}

\begin{thm}\label{maintheorem}
Consider the Cauchy problem for the Dirac equation in the non-extreme 5-dimensional Myers-Perry black hole geometry outside the event horizon

\begin{equation}
(\gamma^A(\partial_A+\Gamma_A)-m)\phi(t,y)=0,\qquad \phi(0,y)=\phi_0(y),
\end{equation}
where  $\phi_0\in L^2((r_+,\infty)\times S^3, d\mu)^4$ and $d\mu$ is the induced measure from the metric (\ref{metric}) on a constant $t$ hypersurface. Then, for any $\delta>0$ and $R>r_++\delta$, defining the compact annulus $K_{\delta,R}=\{r_++\delta\leq r \leq R\}$ and the future directed normal $\nu$, we have

\begin{equation}\label{diraccurrent}
\lim_{t\to\infty} \int_{K_{\delta,R}}\! \left(\overline{\phi}\gamma^j\phi\right)\!(t,y)\nu_j \:d\mu = 0.
\end{equation}
\end{thm}
\begin{proof}
We begin by showing that the transformed spinor $\Psi_0=\Delta^{1/4}\sqrt{r+ip\gamma^5}\phi_0$ is in $\mathcal{H}$. We have 

\begin{equation}
d\mu \sim \frac{r^5}{\sqrt{\Delta}}\sin\theta\cos\theta \: drd\theta d\varphi d\psi, \qquad d\nu = r^4 \sin\theta\cos\theta \: dx d\theta d\varphi d\psi.
\end{equation}

Since $\phi_0\sim \Delta^{-1/4} r^{-1/2}\Psi_0$, it follows from the change of radial variable (\ref{radialvariable}) that $\Psi_0\in\mathcal{H}$.

The rest of the proof follows along the lines of \cite{Finster_2003}; we hereby provide a sketch of the proof. By the density of $C_0^\infty(\mathbb{R}\times S^3)^4\subset \mathcal{H}$, the compactness of the domain $K_{\delta,R}$ and the weak convergence of the integral spectral representation (\ref{intrep}), it is shown in \cite{Finster_2003} that it is sufficient to consider the projection of a smooth test spinor $\Psi_I\in C_0^\infty(\mathbb{R}\times S^3)^4$ onto finitely many angular modes and azimuthal quantum numbers,

\begin{equation}\label{finitesuminitial}
\Psi_{k_{a0},k_{b0},n_0}(y)=\frac{1}{\pi}\sum_{|k_a|\leq k_{a0}} \sum_{|k_b|\leq k_{b0}} \sum_{|n|\leq n_0} \int_{-\infty}^\infty d\omega \sum_{a,b=1}^2 t_{ab}^{k_ak_b\omega n}\Psi_a^{k_ak_b\omega n}(y)\langle \Psi_b^{k_ak_b\omega n}|\Psi_I \rangle.
\end{equation}
Noting that the coefficients $t_{ab}$ are bounded, by Lemma \ref{lem1}, the integrand of (\ref{finitesuminitial}) for a fixed $k_a,k_b,n$,

\begin{equation}\label{knintegrand}
\sum_{a,b=1}^2 t_{ab}^{k_ak_b\omega n}\Psi_a^{k_ak_b\omega n}(y)\langle \Psi_b^{k_ak_b\omega n}|\Psi_I \rangle,
\end{equation}

\noindent is bounded, locally uniformly in $y$ and $\omega$. By the convergence of the integral representation of Theorem \ref{intreptheorem}, (\ref{knintegrand}) is in $L^1(\mathbb{R})^4$ as a function of $\omega$, with an $L^1$ bound locally uniform in $y$. Consider the solution to the Cauchy problem with initial data $\Psi_{k_{a0},k_{b0},n_0}(0,y)=\Psi_{k_{a0},k_{b0},n_0}(y)$,

\begin{equation}\label{finiteknsolution}
\Psi_{k_{a0},k_{b0},n_0}(t,y)=\frac{1}{\pi}\sum_{|k_a|\leq k_{a0}} \sum_{|k_b|\leq k_{b0}} \sum_{|n|\leq n_0} \int_{-\infty}^\infty d\omega \: e^{-i\omega t} \sum_{a,b=1}^2 t_{ab}^{k_ak_b\omega n}\Psi_a^{k_ak_b\omega n}(y)\langle \Psi_b^{k_ak_b\omega n}|\Psi_I \rangle.
\end{equation}

For each $k_a,k_b,n$, the integrand of the $d\omega$ integral in (\ref{finiteknsolution}) is the Fourier transform of (\ref{knintegrand}), and by the above $L^1$ bound it is $L^\infty$ in $t$, locally uniformly in $y$. By the Riemann-Lebesgue lemma \cite{reed_simon_1975}, it tends to zero as $t\rightarrow \infty$ for each $y$. Thus, since the $k_a,k_b,n$ sums in (\ref{finiteknsolution}) are finite, we have

\begin{equation}
\lim_{t\to\infty}\Psi_{k_{a0},k_{b0},n_0}(t,y)=0.
\end{equation}

Since the integrand $\overline{\phi}\gamma^j\phi$ of the Dirac current (\ref{diraccurrent}) is uniformly bounded and converges pointwise in $y$ to zero, the integral (\ref{diraccurrent}) converges to zero by the dominated convergence theorem.

\end{proof}

\section{Perspectives and future work}
The decay result of \cite{Finster_2003} does not hold in the case of the extreme Kerr geometry, in which the mass of the black hole is equal to the magnitude of its angular momentum and the two horizons coincide. In fact, in the exterior region of such a black hole, Schmid \cite{Schmid_2004} shows the existence of time-periodic $L^2$-normalizable bound state solutions to the Dirac equation. It would be of interest to determine whether $L^2$ bound state solutions exist in the extreme 5D Myers-Perry geometry.

The Boyer--Lindquist-type coordinates used in the present paper and in \cite{Finster_2003} are only well-defined outside the event horizon, as they have metric coefficients that are singular at the event and Cauchy horizons. In \cite{R_ken_2017}, Röken derives Eddington--Finkelstein-type coordinates on the Kerr black hole, an analytic extension of the Boyer-Lindquist-type coordinates through the horizons. In \cite{Finster_2016}, Finster and Röken construct a self-adjoint extension of the Dirac Hamiltonian in a class of spacetimes which include the case of spacetimes containing horizons, such as black holes in Eddington--Finkelstein-type coordinates, where the Dirac Hamiltonian is not elliptic. This allows them to derive in \cite{Finster_2018} an integral spectral representation for the Dirac propagator in the Kerr geometry using Stone's formula. Since the self-adjoint extension of \cite{Finster_2016} includes the case of the 5D Myers-Perry geometry in Eddington--Finkelstein-type coordinates, it would be relevant to derive an integral spectral representation for the Dirac propagator through the horizons in 5 dimensions, analogously to what was done in \cite{Finster_2018} in the Kerr geometry.

Using the integral spectral representation they obtained in \cite{Finster_2003} as main technical tool, Finster, Kamran, Smoller and Yau \cite{Finster_2002} derive pointwise decay rates and probability estimates for Dirac spinors in the Kerr-Newman geometry. In particular, they show that for generic initial data, Dirac spinors decay at a rate of $t^{-5/6}$, which is slower than the decay rate of $t^{-3/2}$ in Minkowski space. Using the integral representation (\ref{intrep}) as a starting point, it would be of interest to determine whether analogous results hold in the 5D Myers-Perry geometry.

The results of \cite{Finster_2003} and the present work illustrate considerable similarities between the behaviour of Dirac spinors in 4 and 5-dimensional rotating black hole metrics. It is a promising question to ask whether the radial asymptotic behaviour and local decay hold in higher dimensions. However, separation of variables for the Dirac equation in Myers-Perry metrics of dimension $\geq 6$ has not yet been done; it thus offers an interesting direction for future work.

\medskip

\noindent\textbf{Acknowledgments: } The author is grateful to Prof. Niky Kamran for guidance, advice and careful proofreading of the manuscript. The author also thanks Prof. Felix Finster for comments and suggestions. Finally, the author is grateful to the referees for helpful criticism and suggestions that lead to improvements in the paper. This work was supported by the NSERC Undergraduate Student Research Award program and by NSERC grant RGPIN 105490-2018.
\pagebreak
\appendix

\section{Nondegeneracy and regularity of the angular eigenfunctions}

To justify our choice of basis for $\mathcal{H}_{x_1,x_2}$ as eigenvectors of $\gamma^1\mathcal{A}^{k_ak_b}$, we need the following result.

\begin{pro}

For given $k_a, k_b$ and $\lambda\in\sigma(A)$, there is at most one eigensolution of (\ref{angular}), which we denote $Y^{k_ak_b}$. Both $\lambda$ and $Y^{k_a,k_b}$ depend smoothly on $\omega$ for all $\omega\in\mathbb{R}$.

\end{pro}

\begin{proof}
In $\mathcal{H}^{k_ak_b}_{x_1,x_2}$, the angular ODE $AY=\lambda Y$ is equivalent under the ansatz (\ref{ansatz}) to the eigenvalue problem $\gamma^1\mathcal{A}\Psi=\lambda\Psi$. Since $\gamma^1$ and $\mathcal{A}$ commute and it is known \cite[Appendix A]{Daud__2012} that $\mathcal{A}$ is self-adjoint on $L^2(S^3, d\Omega)^4$ and has compact resolvent, we obtain that $\gamma^1\mathcal{A}$ and therefore $A$ have real, discrete spectra with finite-dimensional eigenspaces.

The two fundamental solutions of (\ref{angular}) behave near $\theta=0$ as 

\begin{equation}\label{fundsols}
Y^{k_ak_b}_1=(\theta^{-k_a-1} + o(\theta^{-k_a-1}), o(\theta^{-k_a-1})), \qquad Y^{k_ak_b}_2=(o(\theta^{k_a}), \theta^{k_a}+o(\theta^{k_a})).
\end{equation}

For each $k_a\in\mathbb{Z}$, one of the solutions $Y^{k_ak_b}_{1/2}$ diverges faster than $\theta^{-1}$ as $\theta\rightarrow 0$. To prove nondegeneracy of the spectrum, we will follow \cite[Appendix A]{Finster_2000} and consider the partial differential equation for  $\alpha(\theta,\varphi)$ given by

\begin{equation}\label{pde}
\left( C_b(\theta)\sigma^3+iL_\theta\sigma^1-\mathscr{C}_a(\theta)\sigma^2\right)\alpha=\lambda\alpha,
\end{equation}

\noindent where $C_b$, $L_\theta$ are as in (\ref{angularabbrevs}) and

\begin{equation}
\mathscr{C}_a(\theta)=-\frac{i}{\sin\theta}\partial_\varphi -\frac{\omega(a^2-b^2)\sin\theta\cos\theta}{p}.
\end{equation}

On the domain $(\theta,\varphi)\in (0,\pi)\times(0,2\pi)$ with the boundary conditions

\begin{equation}
\lim_{\varphi\searrow 0}\alpha(\theta,\varphi)=-\lim_{\varphi\nearrow 2\pi}\alpha(\theta,\varphi)
\end{equation}

\noindent and under the ansatz

\begin{equation}
\alpha(\theta,\varphi)=e^{-i(k_a+\frac{1}{2})\varphi}\begin{pmatrix}
iY_+(\theta)\\
Y_-(\theta)
\end{pmatrix},
\end{equation}
the PDE (\ref{pde}) simplifies to the angular ODE (\ref{angular}). We may write (\ref{pde}) as 

\begin{equation}
\left(i\sigma^1\left(\partial_\theta + \frac{\cot\theta}{2}\right) + i\sigma^2\frac{1}{\sin\theta}\partial_\varphi+\mathcal{A}_0\right)\alpha=\lambda\alpha,
\end{equation}

\noindent where 

\begin{equation}
\mathcal{A}_0=-i\frac{\tan\theta}{2}\sigma^1+\frac{\omega(a^2-b^2)\sin\theta\cos\theta}{p}\sigma^2+\left(-\frac{(k_b+\frac{1}{2})}{\cos\theta}+mp-\frac{\omega ab}{p}\right)\sigma^3
\end{equation}

\noindent is smooth in a neighbourhood of $\theta=0$. We perform the transformation $\tilde{\alpha}=U\alpha$ with

\begin{equation}
U(\theta,\varphi)=\exp\Big(-i\frac{\varphi}{2}\sigma^3\Big)\exp\Big(-i\frac{\theta}{2}\sigma^2\Big).
\end{equation}

As in \cite[Appendix A]{Finster_2000}, the equation $\tilde{\mathcal{A}}\tilde{\alpha}=\lambda\tilde{\alpha}$ thus obtained can be seen as an eigenvalue problem on $S^2$. However, due to the singularities of $\mathcal{A}_0$ at $\theta=\frac{\pi}{2}$, it is only smooth in some neighbourhood $V\subset S^2$ of $\theta=0$. In fact $\tilde{\mathcal{A}}$ is first order and elliptic in $V$, and thus by the elliptic regularity theorem \cite{taylor_2011}, its eigensolutions $\tilde{\alpha}$ are in the Sobolev space $H^1(V)\subset L^2(V, d\Omega)$. However, choosing $\epsilon>0$ such that $V_\epsilon = \{(\theta,\varphi)\in S^2|\theta<\epsilon\}\subset V$, we notice that for all $l\geq 1$,

\begin{equation}
\left\Vert\frac{1}{\theta^l}\right\Vert^2_{L^2(V,d\Omega)}=\int_{V}\frac{1}{\theta^{2l}}d\Omega\geq \int_{V_\epsilon}\frac{1}{\theta^{2l}}d\Omega=2\pi\int_0^{\epsilon}\frac{\sin\theta}{\theta^{2l}} d\theta = \infty.
\end{equation}

We can therefore rule out one of the two fundamental solutions (\ref{fundsols}) and thus conclude that the spectrum of $A$ is nondegenerate.

As $A=A(\omega)$ is an analytic family of operators with nondegenerate spectra, by the Kato-Rellich theorem \cite{reed_simon_1978} its eigenvalues $\lambda(\omega)$ and eigenvectors $Y^{k_ak_b}(\omega)$ depend smoothly on $\omega$. 

\end{proof}

\pagebreak
\bibliography{refs}
\bibliographystyle{amsplain}
\end{document}